\documentclass[12pt]{article}

\usepackage[T1]{fontenc}
\usepackage{authblk}
\usepackage{amsthm}
\usepackage{newtxtext,newtxmath}
\usepackage{microtype}
\usepackage{graphicx}
\usepackage[noend]{algpseudocode}
\usepackage{fullpage}
\usepackage[font=small,labelfont=bf]{caption}

\newtheorem{theorem}{Theorem}
\newtheorem{algorithm}[theorem]{Algorithm}
\newtheorem{definition}[theorem]{Definition}
\newtheorem{claim}[theorem]{Claim}

\DeclareMathAlphabet{\altmathcal}{OMS}{cmsy}{m}{n}
\newcommand{\rss}{z}

\title{An Exact, Linear Time Barabási--Albert Algorithm}
\date{}
\author[1]{\large Giorgos Stamatelatos}
\author[1]{\large Pavlos S. Efraimidis}
\affil[1]{\normalsize Dept. of Electrical and Computer Engineering, Democritus University of Thrace, Kimmeria, Xanthi 67100, Greece}

\begin{document}

\maketitle

\begin{abstract}
This paper presents the development of a new class of algorithms that accurately implement the preferential attachment mechanism of the Barabási--Albert (BA) model to generate scale-free graphs. Contrary to existing approximate preferential attachment schemes, our methods are exact in terms of the proportionality of the vertex selection probabilities to their degree and run in linear time with respect to the order of the generated graph. Our algorithms utilize a series of precise, diverse, weighted and unweighted random sampling steps to engineer the desired properties of the graph generator. We analytically show that they obey the definition of the original BA model that generates scale-free graphs and discuss their higher-order properties. The proposed methods additionally include options to manipulate one dimension of control over the joint inclusion of groups of vertices.

\bigskip\noindent\small\textit{Keywords}: Preferential Attachment, Scale-Free Graphs, Barabási--Albert Model, Random Sampling
\end{abstract}

\section{Introduction}


The Barabási--Albert (BA) model~\cite{barabasi1999emergence} is a growing preferential attachment mechanism that dictates the rules of connections among vertices when newborn nodes enter the network. Specifically, it requires selecting $m$ vertices from the graph population when a newborn node enters the network in such a way that the probability of selecting a vertex is proportional to its degree. By repeating this process, the model results in the generation of a scale-free graph of order $n$. According to the analytical arguments made in the original paper, the random sampling model involved in this growing process is a selection without replacement and with inclusion probabilities strictly proportional to the degree of the vertices. Consequently, a round of the BA model is defined as follows:

\begin{definition}\label{def:ba}
(A round of the BA model.) Given a graph of order at least $m$, $m$ distinct nodes are randomly selected from the vertex population. The probability of each node to be present in the sample is exactly proportional to its degree. A new node is inserted into the graph and is connected to each of the $m$ selected nodes.
\end{definition}

Current state of the art preferential attachment models are typically efficient
in terms of their running time but they are not equivalent to the exact
model described by Barabási and Albert; they do not refer to the same simple graph without multiple edges, or the probability model employed is only an approximation of the original scheme.


One of the most popular theoretical
models in the literature is the model of Bollobás~\cite{bollobas2001degree}, which employs a probability model that guarantees strict proportionality but results in a multigraph. Due to its simplicity and the rigorous analysis made in this work about various properties of this network, the model has since been adopted in the literature. Another example of a multigraph generator is mentioned in~\cite[Chapter 8]{van2016random}, where the edges are added with intermediate weight updates with replacement, a scheme that results in possible multiple edges per node pair.


Other studies correctly treat the BA model as a simple graph but with a probability selection scheme that is only an approximation of the original model. For example, Hadian et al~\cite{hadian2016roll} define a simple graph generator but the probabilities of node inclusions are not exactly proportional to their degree due to rejections. This difference has been explained further in~\cite{stamatelatos2021weighted}, where the distinction of inclusion and selection probabilities is made. Berger et al~\cite{berger2014asymptotic} attempt to make a distinction about different probability schemes (independent, conditional, sequential) but also refer to the Bollobás multigraph model.


Despite the models mentioned previously being both efficient and rigorously studied, they do not strictly abide by the definition of the original BA model. The definition requires a sampling scheme without replacement that generates simple graphs, and the inclusion probability of a vertex is strictly proportional to its degree. These requirements can be summarized into the \textit{str$\pi$ps} random sampling scheme, which refers to a weighted random sampling design without replacement with inclusion probabilities strictly proportional to degree.
We give the definition of the str$\pi$ps scheme due to the term str$\pi$ps not being consistent throughout the literature:

\begin{definition}
(The str$\pi$ps random sampling scheme.) Assume a weighted random sampling scheme $S$ that randomly selects exactly $k$ discrete items from a population of $n$ items with weights $w_i$, $i=1,2,\ldots,n$. Then $S$ satisfies the str$\pi$ps property if the inclusion probability $\pi_i$ of element $i$ is exactly proportional to $w_i$ or is equal to $k \cdot (w_i / \sum_j w_j)$ for all $i$.
\end{definition}


In this paper, we present a class of algorithms that obey the definition of the BA model strictly, both with respect to the type of the output graph and the interpretation of the probabilities being employed in the preferential attachment step being exactly proportional to the node degrees. Our algorithms also run in linear time with respect to $n$, or in constant time for each time step, i.e., for each new node. It is worth noting that it is trivial to apply any \textit{str$\pi$ps} sampling method on each time step independently but that would result in quadratic complexity for the whole process.
The computational complexity of the algorithm is of great importance since the sizes of the generated scale-free graphs are often very large, up to hundreds of thousands or even millions of nodes. To our knowledge, up to now there was no efficient sampling algorithm for running each step in $\altmathcal{O}(1)$ with respect to $n$. This could be a reason why current software libraries (e.g. networkx, igraph) have opted for a related efficient weighted random scheme that only approximates the BA definition. To the best of our knowledge this is the first time that both a strict interpretation of the inclusion probabilities and a linear running time are satisfied.

The starting point of our methods is the algorithm described in~\cite[Section 9.3]{dorogovtsev2002evolution}. During each time step, a random edge is selected and the newborn node is connected with the ends of that edge creating a triangle. The process is repeated until the graph is of the desired size. It is easy to show that on this sampling scheme, where each edge is guaranteed to connect two different nodes, each vertex exists in the edge set as many times as its degree and, therefore, its inclusion probability on any time step is exactly proportional to its degree, satisfying the \textit{str$\pi$ps} scheme. The same method was also mentioned previously in~\cite[Section IV]{batagelj2005efficient}:
\begin{quote}
    in a list of all edges created thus far, the number of occurrences of a vertex is equal to its degree, so that it can be used as a pool to sample from the degree-skewed distribution in constant time
\end{quote}
Here, we note that, regarding the joint inclusion of the vertices, the described model is only capable of sampling adjacent pairs of nodes.

This concept of random sampling where the sample space is computed and maintained ahead of time so that a random sample can be generated in constant time is known as \textit{whole sampling}~\cite[Section 1.7.1]{brewer2013sampling}. In this work, the approach of whole sampling is used to prepare a pool of sets of nodes which is then used on another sampling algorithm to generate the random sample of $m$ nodes. The result is a more generalized, advanced and efficient class of algorithms that accurately implement the BA model. Moreover, the technique of whole sampling, which is a central component in our proposed algorithms, is a new perspective of the algorithmic study of preferential attachment, which until now is based on consecutive node selections to execute a time step.

In particular, using this simple $m=2$ model as an entry point, we first analytically formulate its local clustering coefficient properties based on the observation that every newborn node will create one triangle with probability 100\%. Additionally, we propose the improvement of the algorithm with an additional step that allows the selection of any pair of nodes at any time step and overcomes the limitation of the algorithm only being capable of sampling adjacent nodes while at the same time does not asymptotically affect its complexity. Finally, we propose a new class of algorithms that generalize the operation for $m>2$ using a series of additional, carefully designed and precise random sampling steps, while also offering options for joint node selection guarantees. Specifically, we take advantage of the operation of Jessen's whole sampling algorithm~\cite{jessen1969some}, a weighted \textit{str$\pi$ps} random sampling scheme, and adjust its operation to fit the preferential attachment problem by introducing auxiliary data structures. Similarly to the introductory $m=2$ model where selections of pairs of nodes are made, our generalization allows the selection of sets of $m$ nodes by adapting the list of edges into a list of hyperedges of an implicit $m$-uniform hypergraph. We prove that our new approaches satisfy the \textit{str$\pi$ps} model and show their running time to be linear with respect to $n$.


Our contribution can be summarized as follows:
\begin{enumerate}
    \item Analytical findings about the clustering properties of the simple $m=2$ case model.
    \item The generalization of the same model to allow sampling of all pairs of nodes instead of only adjacent nodes.
    \item Efficient and accurate generalizations for $m>2$ that implement the definition of the BA model strictly and run in linear time with respect to $n$.
\end{enumerate}


We begin the presentation with algorithm SE-A (Section~\ref{sec:algorithm-1}), which implements the $m=2$ model present in the literature with the additional steps required to allow sampling of any pair of nodes.
A generalization is then given for $m>2$ as the abstract algorithm SE (Section~\ref{sec:algorithm-abstract}) that utilizes an auxiliary data structure $H$ that resembles a $m$-uniform hypergraph. Due to the generalization being abstract, we propose two possible implementations that achieve the generation of scale-free graphs in slightly different ways. Algorithm SE-B (Section~\ref{sec:algorithm-2}) is a greedy approach
towards the growth of the $H$ data structure while its variation SE-B* additionally ensures that every $m$-combination of nodes has non-zero probability of being jointly selected at any step of the generation process. Algorithm SE-C (Section~\ref{sec:algorithm-3}) is an alternative, more systematic approach towards the updating of the $H$ data structure that demonstrates the relation between preferential attachment and the random sampling problem, while still maintaining linear complexity. Algorithms SE-B, SE-B* and SE-C reduce to algorithm SE-A when $m=2$. We prove that both algorithms are correct with respect to the \textit{str$\pi$ps} probability model and have linear time worst case complexities.
Each algorithm is described as a growing process starting from an initial graph $\altmathcal{G}_0$ until it reaches the desired order $n$. For simplicity, it is assumed that the starting graph is connected, otherwise unconnected graphs may be generated. A discussion surrounding the properties of initial graphs is given in Section~\ref{sec:initial-graph}.
Finally, it is worth mentioning that the case of $m=1$ is not discussed here as there is no distinction between the interpretation of the proportionality of the probabilities and, hence, existing mechanisms obey the \textit{str$\pi$ps} model for this particular case.

\section{Algorithm SE-A}
\label{sec:algorithm-1}


Algorithm SE-A is the introductory version of our approach, works only for $m=2$ and accepts a parameter $z$. For the trivial case of $z=1$, Algorithm SE-A reduces to the algorithm described in~\cite[Section IV]{batagelj2005efficient} and~\cite[Section 9.3]{dorogovtsev2002evolution}.
Instead of sampling a single edge of the existing graph, we perform a random unweighted selection with replacement of $\rss \ge 1$ edges, and use the union of vertices in those edges along with their frequencies in another sampling scheme. The result is the vertices $u,w$ that receive edges from the newborn node $v$.
A high level sketch of SE-A appears in Algorithm~\ref{alg:algorithm-1}. All following algorithms are assumed to have an implicit random number generator input.

\begin{algorithm} Algorithm SE-A (high level sketch) \\
\noindent\hspace*{\algorithmicindent} \textbf{Input} An initial connected graph $\altmathcal{G}_0(V_0,E_0)$ containing at least one edge, the desired number of vertices $n$, and the parameter $\rss$ \\
\noindent\hspace*{\algorithmicindent} \textbf{Output} A scale-free graph $\altmathcal{G}(V,E)$
\begin{algorithmic}[1]
    \State $(V,E) \gets (V_0,E_0)$
    \While{$|V| < n$}
        \State select $\rss$ uniformly random edges $E_r=(e_1,e_2,\ldots,e_{\rss})$ with replacement from $E$
        \State select $(u,w)$ using RSS from the union of vertices in $E_r$ and their frequencies
        \State add new vertex $v$ to $V$
        \State add new edges $(v,u)$ and $(v,w)$ to $E$
    \EndWhile
    \State \textbf{return} $\altmathcal{G}(V,E)$
\end{algorithmic}
\label{alg:algorithm-1}
\end{algorithm}

The additional sampling scheme can be any algorithm satisfying the str$\pi$ps property. Here, our analysis makes use of the random systematic sampling (RSS) algorithm~\cite{goodman1950controlled}, which is inline with other utilized methods in our presentation and additionally guarantees that every combination of vertices in the input has a non-zero inclusion probability due to the intermediate shuffling step. RSS is shown in Algorithm~\ref{alg:random-systematic-sampling}. The description considers integral weights only. The notation $\langle 1, n \rangle$ in this algorithm and throughout this work denotes the closed integer sequence $1,2,\ldots,n$.

\begin{algorithm} Random Systematic Sampling Algorithm \\
\noindent\hspace*{\algorithmicindent} \textbf{Input} A bag of elements $x_1, x_2, \dots, x_n$ where each element $i$ appears with frequency $d_i$, the sum of all frequencies $s = \sum d_i$, and the desired number of elements in the sample $m$, where $m$ divides $s$ \\
\noindent\hspace*{\algorithmicindent} \textbf{Output} A random sample of $m$ discrete elements
\begin{algorithmic}[1]
    \State shuffle $x$
    \State $psum \gets 0$
    \State generate uniformly random variate $r$ in $[0,s/m)$
    \ForAll{$i \in \langle 1,n \rangle$}
        \State $psum$ += $d_i$
        \If{$psum > r$}
            \State \textbf{yield} $x_i$
            \State $r$ += $s/m$
        \EndIf
    \EndFor
\end{algorithmic}
\label{alg:random-systematic-sampling}
\end{algorithm}

The intuition of this algorithm is that a vertex exists in the edge set as many times as its degree and no edge can contain the same vertex more than once, a property which guarantees the correctness of the algorithm with respect to Definition~\ref{def:ba}. Specifically, Algorithm SE-A is reduced to the simple algorithm in~\cite{batagelj2005efficient} and~\cite{dorogovtsev2002evolution} for $\rss = 1$, since the RSS step is equivalent to a uniformly random edge selection. As a result, during a single uniform random edge selection, the inclusion probability of a vertex with degree $d$ at any time is $d/|E|$. Therefore, the probability of a vertex with degree $d$ gaining an edge at any time after a newborn node has been added is exactly proportional to its degree, which leads to the following theorem regarding the correctness of the algorithm:

\begin{theorem}\label{theorem:sea-correctness}
    Algorithm SE-A satisfies the \textit{str$\pi$ps} probability scheme and generates a simple graph according to the definition of the BA model.
\end{theorem}

\begin{proof}
For a node $v$ with degree $d$, the parameters $m$ and $\rss$, and the current number of edges $|E|$, the probability of a uniformly random edge containing $v$ is $p = d/|E|$. Therefore, the probability of selecting $v$ $k$ times by making $\rss$ unweighted selections with replacement during the RSS step is
\[
    f(k) = \binom{\rss}{k} p^k (1-p)^{\rss-k}.
\]

The probability $g(k)$ that a vertex selected $k$ times to appear in the final sample after the completion of the RSS step is exactly proportional to $k$ because RSS satisfies the str$\pi$ps property. Hence,
\[
    g(k) = \frac{k}{\rss m} m = \frac{k}{\rss}.
\]

Combining the previous results, the probability $\pi_v$ of the node $v$ to appear in the final sample is
\[
    \pi_v = \sum_{k=0}^{\rss} g(k) f(k) = \sum_{k=0}^{\rss} \frac{k}{\rss} \binom{\rss}{k} p^k (1-p)^{\rss-k} = \frac{1}{\rss} \sum_{k=0}^{\rss} k \binom{\rss}{k} p^k (1-p)^{\rss-k}.
\]
The quantity can be simplified (expectation of the binomial distribution) into
\[
    \pi_v = \frac{1}{\rss} \rss p = p = \frac{d}{|E|},
\]
which is exactly proportional to the degree $d$.
\end{proof}

\subsection{Complexity}

The work performed by algorithm SE-A during its growth function is $\rss$ random edge selections, at most $2\rss$ iterations during the RSS step, one vertex and two edge additions, leading to the following theorem:

\begin{theorem}
    Algorithm SE-A runs in time $\altmathcal{O}(n \rss)$ to create a graph of $n$ vertices using a selection of $\rss$ edges during the RSS step.
\end{theorem}

Algorithm SE-A does not require any additional memory other than the output graph and doesn't use auxiliary data structures.

\subsection{Clustering Coefficient of the $\rss = 1$ case}

Algorithm SE-A is a general case of the simple model described in~\cite[Section 9.3]{dorogovtsev2002evolution} and is equivalent to it for $\rss = 1$ and considering an initial graph of two connected vertices. For this case ($z=1$ and initial graph of two connected vertices), in this part we study several properties surrounding the local clustering coefficient.
First, we state the following theorem regarding the correlation of the local clustering coefficient with the degree:

\begin{theorem}\label{theorem:local-clustering}
    The local clustering coefficient of a vertex with degree $d$ is $2/d$ at any time of the generation process.
\end{theorem}
\begin{proof}
    The local clustering coefficient of a vertex $i$ at time $t$ is
    \begin{equation}\label{eq:local-clustering-formula}
        C(d_i) = \frac{2 E_i}{d_i (d_i - 1)},
    \end{equation}
    where $E_i$ is the amount of edges among $i$'s neighbors and $d_i$ the degree of \(i\). Therefore, when $i$ enters the network, its local clustering coefficient is 1, because $d=2$ and $E_i=1$. If $i$ does not acquire any new edges, its local clustering coefficient will not change, because none of the quantities $E_i$ or $d_i$ will change, since no edges are created among existing vertices. If $i$ obtains one edge, both its degree and $E_i$ will be increased by 1, because new vertices only connect to existing edge's endpoints. Thus, $d_i$ and $E_i$ are connected via the formula $d_i=E_i+1$ for every $i$ that enters the network after its initial formation. Replacing in (\ref{eq:local-clustering-formula}) gives
    \begin{equation}\label{eq:clustering-degree}
        C(d) = \frac{2 (d - 1)}{d (d - 1)} = \frac{2}{d}.
    \end{equation}
    It is easy to see that the formula also holds for the two vertices in the initial graph after the third node has entered.
\end{proof}

Theorem~\ref{theorem:local-clustering} shows that the local clustering coefficient of a vertex depends only on its degree. As a result, the limit of the average local clustering coefficient can now be derived based on the expected power law degree distribution:

\begin{theorem}\label{theorem:average-local-clustering}
    The limiting average local clustering coefficient of a graph produced using the SE-A ($\rss = 1$) algorithm is $2 \pi^2 - 19$.
\end{theorem}
\begin{proof}
    The degree distribution of a graph generated using the \textit{str$\pi$ps} method when $n \to \infty$ is given in~\cite[Equation 90]{albert2002statistical}:
    \begin{equation}\label{eq:degree-distribution}
        P(d) = \frac{2m(m+1)}{d(d+1)(d+2)}, d \ge m.
    \end{equation}
    By combining (\ref{eq:clustering-degree}) and (\ref{eq:degree-distribution}) and setting $m=2$, we can get the average local clustering coefficient:
    \begin{equation}\label{eq:se-a-average-clustering}
        C_{avg} = \sum_{d=2}^{\infty} (2/d) P(d) = \sum_{d=2}^{\infty} \frac{24}{d^2(d+1)(d+2)} = 2 \pi^2 - 19 \approx 0.73921.
    \end{equation}
\end{proof}

Additionally, the local clustering coefficient distribution can also be formulated. Exchanging $d$ with $2/c$ and setting $m=2$ in Equation~\ref{eq:degree-distribution}, we get the local clustering distribution $P_c$:
\begin{equation}\label{eq:lcc-distribution}
    P_c(c) = P(2/c) = \frac{6c}{(2/c+1)(2/c+2)}.
\end{equation}

The variance of the local clustering distribution is finite and equal to
\begin{align*}
    \sigma^2 &= E[C^2] - (E[C])^2 \\
             &= \sum_{d=2}^{\infty} \left( 2/d \right)^2 P(d) - C_{avg}^2 \\
             &= 24 \zeta(3) - 330 + 70 \pi^2 - 4 \pi^4 \\
             &\approx 0.08531.
\end{align*}

The value of the average limiting local clustering coefficient produced by algorithm SE-A ($\rss = 1$) is inline with empirical observations of real social networks~\cite[Table~I]{albert2002statistical}. Thus, the SE-A mechanism can simulate features of real social networks beyond the power law degree distribution. The operation of the algorithm can also be parallelized with~\cite{holme2002growing}, where the clustering coefficient is tunable by setting a probability of creating a triangle when new vertices enter the network. Here, this probability is $1$ because the new vertex always connects with two endpoints of the same edge. As a result, the number of triangles in the final graph $\altmathcal{G}$ is $tri(\altmathcal{G}) = tri(\altmathcal{G}_0) + n - 2 \sim n$. Here the tilde symbol is set to mean asymptotic equivalence~\cite[Section~1.4]{de1981asymptotic}. It is worth noting that for other values of $\rss$, the probability of triangle formation is lower than 1.

\subsection{Joint inclusion probabilities and the parameter $\rss$}

While the BA model does not explicitly specify the desired higher-order inclusion probabilities (joint probabilities) of nodes to be selected, such discussion emerges from the operation of Algorithm SE-A. In particular, it is clear that the simple version of SE-A ($\rss = 1$) is only able to sample adjacent pairs of nodes such that the joint inclusion probability of a pair of vertices $(u,w)$ is given by
\begin{equation*}
    \pi_{uw}=
    \begin{cases}
      1/|E|, & \text{if $u$ and $w$ are adjacent} \\
      0, & \text{otherwise.}
    \end{cases}
\end{equation*}
This situation suggests that the joint selections are heavily biased towards the structure of the graph and is the reason why this model results in high clustering graphs. Considering that there exist more than one combinations of higher order inclusion probabilities that correspond to the same strict first order inclusion probabilities, this observation raises the interesting question
\begin{quote}
    For two vertices $u$ and $w$ with degrees $d_u$ and $d_w$, what should be the value of the second order inclusion probability $\pi_{uw}$?
\end{quote}

Although answering this question is outside the scope of this paper, it might be challenging to provide a single universally correct answer. For example, one natural interpretation that eliminates the bias with respect to the graph structure would be the instance where all conditional inclusion probabilities are uncorrelated with the other elements of the sample, i.e. for all combinations of vertices $u,w$, the inclusion probability of $u$ is independent of whether $w$ is in the sample, or $\pi_{u|w} = \pi_u$, hence $\pi_{uw} = \pi_u \pi_w$. However, this property is impossible to be guaranteed for all pairs of vertices in the str$\pi$ps scheme and for any graph that can be generated by the algorithm. We show this via fundamental properties of random sampling without replacement next.

For $m=2$, consider a graph that corresponds to the first order inclusion probabilities $\pi_1$, $\pi_2$, \ldots, $\pi_n$. Then, if we accept that $\pi_{uw} = \pi_u \pi_w$ for all pairs of $u,w$, the inclusion probability of node $i$ is given by
\begin{equation*}
    \pi_i = \sum_{j \neq i} \pi_{ij} = \sum_{j \neq i} \pi_i \pi_j = \pi_i \sum_{j \neq i} \pi_j.
\end{equation*}
Considering that $\sum_{i} \pi_i = m = 2$, it follows that $2 - \pi_i = 1$ for every $i$, implying the single solution of $n=2$ and $\pi_1 = \pi_2 = 1$, disproving general validity.

Equivalently, for $m=3$, the first order inclusion probability of node $i$ corresponds to the equation
\begin{equation*}
    \pi_i = \sum_{p \neq i, q \neq i, q > p} \pi_p \pi_q \pi_i \text{~~~~~~or~~~~~~} \sum_{p \neq i, q \neq i, q > p} \pi_p \pi_q = 1.
\end{equation*}
Subtracting the pair of equations corresponding to nodes $j$ and $k$ and simplifying results in
\begin{equation}\label{eq:disprove-m-3}
    \sum_{q,q \neq j, q \neq k} \pi_j \pi_q = \sum_{q,q \neq j, q \neq k} \pi_k \pi_q,
\end{equation}
implying that all first order inclusion probabilities are equal (all nodes of a graph having exactly the same degree), disproving generality again. The same argument can be extended for $m>3$ via symmetry.

Setting a constant factor $c$ such that $\pi_{uvw} = c \pi_u \pi_v \pi_w$ as an alternative interpretation of lack of bias is infeasible too as Equation~{\ref{eq:disprove-m-3}} remains unchanged (subtraction of constants).

It is worth noting that this property is inherent to any random sampling scheme with the str$\pi$ps property and, therefore, any growing generator that respects Definition~\ref{def:ba}. In contrast, multigraph models do not have this limitation and can be designed to be unbiased. For example the algorithm described in~\cite{batagelj2005efficient} is unbiased with respect to the graph structure since the random draws are independent, albeit contradicting the original definition of the BA model requiring selection without replacement (quoted from~\cite{barabasi1999emergence}):

\begin{quote}
    To incorporate the growing character of the network, starting with a small number ($m_0$) of vertices, at every timestep we add a new vertex with $m$ ($\le m_0$) edges that link the new vertex to $m$ \textit{different} vertices already present in the system.
\end{quote}

Considering the difficulty of defining the desired joint inclusion probabilities, our methods provide a minimal amount of control over the high order probabilities via the RSS step and the parameter $\rss$. In general, the higher the value of $z$, the lower the effect of the graph structure will be in the joint inclusion probability values. More specifically, the RSS step with $\rss \ge 2$ guarantees that every pair of nodes is possible to be jointly selected at any step of the generation process with a non-zero inclusion probability, resulting from the following theorem:

\begin{theorem}
The lower bound of the joint inclusion probability $\pi_{uw}$ of any pair $u,w,u \neq w$ for Algorithm SE-A with $\rss = 2$ and current graph $\altmathcal{G}(V,E)$ is $2/(3|E|^2)$.
\end{theorem}

\begin{proof}
The lower bound is defined by two vertices $u$ and $w$ which have degree 2 and occupy the edges $(u,a)$, $(u,b)$, $(w,a)$, $(w,b)$, noted as $E_1$, $E_2$, $E_3$ and $E_4$. The edge combinations that can result in the joint selection of $u$ and $w$ are $(E_1,E_4)$ and $(E_2,E_3)$. For each combination, the pair $(u,w)$ appears with probability $1/6$ and there are 2 permutations for each of the two combinations. As a result, the lower bound is $2 \cdot (2/|E|^2) \cdot (1/6) = 2/(3|E|^2)$.
\end{proof}

Similarly, for higher values of $\rss$, the $(u,w)$ pair has a strictly positive inclusion probability, since there is at least one selection of $\rss$ edges with replacement that results in such a property with non-zero probability, for example the combination of $E_1$, $E_4$, with any other edge(s). Additionally, the parameter $\rss$ indirectly provides a certain degree of flexibility over the distribution of the joint inclusion probabilities. The relation of this parameter with the joint probability distribution should be investigated in the future.

\section{Generalized Abstract Algorithm SE}
\label{sec:algorithm-abstract}

The main result of this work is an algorithm that runs in linear time and implements the BA model precisely for any constant $m \ge 2$. The algorithm generalizes and extends the ideas presented with Algorithm SE-A and comprises a set of sampling techniques, including a whole sampling method which is due to Jessen~\cite{jessen1969some}.
Jessen's method builds a sample space (called \textit{tableau}) according to the given inclusion probabilities as set of $m$ elements in an iterative way. Each element is assigned a balance quantity proportional to its inclusion probability which is reduced each time the element is used in a sample; the method terminates when all balances are depleted. It is then possible to select one \textit{str$\pi$ps} sample of $m$ elements in constant time.

Here, we exploit the constant time selection and the growing nature of Jessen's method to define the abstract algorithm SE. Algorithm SE maintains a tableau of possible samples as an auxiliary data structure $H$, which comprises a list of sets of $m$ elements such that each node exists in as many sets as its degree. Updating this data structure when newborn nodes enter the network can be performed by increasing the balance of the newborn node and the selected vertices in the tableau without having to repeat the process. We note that for algorithm SE-A, the $H$ data structure is equivalent to the edge set of the network and, hence, not required concretely. The $H$ data structure resembles a $m$-uniform hypergraph, where each of the sets of the list represents one hyperedge. The nature of the process allows multiple copies of the same hyperedge in $H$, similarly to Jessen's method allowing for the same row in the tableau. As a result, $H$ may represent a non-simple hypergraph where repeated edges are possible. Note that even though the hypergraph is non-simple, the $m$-uniform property assures that each hyperedge contains exactly m distinct vertices. Therefore, no loops are permitted. In the rest of this work, we refer to $H$ and its contents in the hypergraph terminology. A high level sketch of this abstract algorithm appears in Algorithm~\ref{alg:algorithm-abstract}.

\begin{algorithm} Abstract Algorithm SE (high level sketch) \\
\noindent\hspace*{\algorithmicindent} \textbf{Input} An initial connected graph $\altmathcal{G}_0(V_0,E_0)$, the desired number of vertices $n$, the desired number of edges added per step $m$, and the parameter $\rss$ \\
\noindent\hspace*{\algorithmicindent} \textbf{Output} A scale-free graph $\altmathcal{G}(V,E)$
\begin{algorithmic}[1]
    \State $(V,E) \gets (V_0,E_0)$
    \State \textit{init}($H$, $\altmathcal{G}_0$) \Comment{initialize $H$ based on $G_0$}
    \While{$|V| < n$}
        \State select $\rss$ uniformly random hyperedges $H_r=(h_1,h_2,\ldots,h_{\rss})$ with replacement from $H$
        \State select $u=(u_1,u_2,\dots,u_m)$ using RSS from the union of vertices in $H_r$ and their frequencies
        \State add new vertex $v$ to $V$
        \State add new edges $(v,u_1),(v,u_2),\dots,(v,u_m)$ to $E$
        \State \textit{update}($H$, $v$, $u$) \Comment{update $H$ based on $v$ and the contents of $u$}
    \EndWhile
    \State \textbf{return} $\altmathcal{G}(V,E)$
\end{algorithmic}
\label{alg:algorithm-abstract}
\end{algorithm}

Algorithm SE, as defined here, is abstract with respect to the \textit{update} function, which can be implemented in numerous ways. This function corresponds to the maintenance of the $H$ hyperedge set when newborn nodes enter the network and is required to satisfy two invariants after the \textit{update} function returns:
\begin{description}
  \item[Invariant~1]
      Each vertex may only exist at most once in a hyperedge.
  \item[Invariant~2]
      Each vertex participates in as many hyperedges as its degree in $\altmathcal{G}$.
\end{description}
These invariants guarantee the correctness of any algorithm, as they are the only necessary conditions for the sampling scheme to be \textit{str$\pi$ps}. It should be noted that the proof of Theorem~\ref{theorem:sea-correctness} can be trivially adjusted for the case of $m>2$. These requirements are implicitly satisfied in the SE-A algorithm since the $H$ data structure is identical to the edge set. The invariants can be simplified by stating that
\begin{quote}
    $H$ must be a (possibly non-simple) $m$-uniform hypergraph with the same degree sequence as $\altmathcal{G}$.
\end{quote}

A general operation of the \textit{update} function is to handle the updating of the $H$ data structure based on the newborn node addition. In particular, the vertices $u_1, u_2, \ldots, u_m$ and $m$ copies of the newborn node $v$ have to be added in $H$. These $2m$ items imply the addition of 2 new hyperedges in $H$. Since no more than 2 copies of $v$ can be added in 2 hyperedges, previously added hyperedges need to be adjusted as well to satisfy the invariants. Two possible methods of achieving this are described in Sections~\ref{sec:algorithm-2} and~\ref{sec:algorithm-3}.

The \textit{init} function represents the initialization of the $H$ data structure so that the invariants are satisfied during the start of the process. Directly following the definition of the invariants, it can be seen that the initial graph $\altmathcal{G}_0(V_0,E_0)$ needs to satisfy the divisibility $2|E_0|/m$ and no vertex can have degree higher than $2|E_0|/m$ for the $H$ data structure to be feasible. The requirements of the initial graph $\altmathcal{G}_0$, which are omitted from Algorithm~\ref{alg:algorithm-abstract} for brevity, are discussed in more detail in Section~\ref{sec:initial-graph}.

The \textit{init} function is marked as abstract because it is possible to be implemented in various different ways. Here, we describe one possible implementation that distributes the vertices in $V_0$ randomly throughout $H$. The method, which we call \textit{random systematic partitioning}, has been influenced by the random systematic sampling method~\cite{goodman1950controlled}, which we here adjust to partition the items instead of sampling them. This algorithm, which -- to our knowledge -- does not seem to have been described before in the literature, might be of independent interest. Random systematic partitioning accepts a bag of elements $x_1, x_2, \ldots, x_n$ where each element $x_i$ is characterized by its frequency $d_i$ (the degree in this context). The goal of the algorithm is to partition the bag into $s$ sets of $m$ elements, where no vertex can exist more than once in any of those $s$ sets. For a feasible outcome, it must hold that $s \cdot m = \sum_{i=1}^n d_i$ and the maximum frequency cannot be higher than $s$. If the frequencies are not known in advance they can be created in one pass over the population.

A high level sketch of random systematic partitioning can be seen in Algorithm~\ref{alg:random-systematic-partitioning}. The algorithm initially shuffles the unique $x$ values and expands them into their frequencies into an implicit $s \times m$ matrix written by row. The output of the operation is the transpose of this matrix; each of the $s$ rows represents one group of the partition. The computational complexity of random systematic partitioning is $\Theta(sm)$ because of the encapsulated loops and is formalized in the following theorem:

\begin{theorem}\label{theorem:random-systematic-partitioning-complexity}
    Given a multiset with $s \cdot m$ items with at most $s$ repetitions of each item, random systematic partitioning runs in time $\Theta(s m)$, and partitions the multiset into $s$ sets of $m$ distinct items each.
\end{theorem}

\begin{algorithm} Random Systematic Partitioning Algorithm (high level sketch) \\
\noindent\hspace*{\algorithmicindent} \textbf{Input} A bag of elements $x_1, x_2, \dots, x_n$ where each element $i$ appears with frequency $d_i$, the desired number of sets $s$, and the desired number of elements in a set $m$ \\
\noindent\hspace*{\algorithmicindent} \textbf{Output} A partitioning $H$ of the input bag into $s$ sets of \(m\) elements randomly distributed across the data structure
\begin{algorithmic}[1]
    \State $H \gets$ array of $s$ empty sets
    \State $k \gets 1$
    \State shuffle $x$
    \ForAll{$i \in \langle 1,n \rangle$}
        \ForAll{$\_ \in \langle 1, d_i \rangle$} \Comment{Perform $d_i$ times}
            \State add $i$ to $H_k$
            \State $k \gets (k \mod s) + 1$
        \EndFor
    \EndFor
    \State \textbf{return} $H$
\end{algorithmic}
\label{alg:random-systematic-partitioning}
\end{algorithm}

\section{Algorithm SE-B}
\label{sec:algorithm-2}

Algorithm SE-B is an implementation of algorithm SE with a minimal approach into implementing the \textit{update} function. The main issue of distributing the $m$ copies of the newborn node $v$ is addressed by inserting one copy into each of the 2 new hyperedges and $m-2$ copies into previously inserted hyperedges. In the latter case, one node from each of these $m-2$ is swapped back into the new hyperedges. A high level sketch of the \textit{update} function of algorithm SE-B is shown in Algorithm~\ref{alg:algorithm-2}.

\begin{algorithm} Algorithm SE-B -- update function (high level sketch) \\
\noindent\hspace*{\algorithmicindent} \textbf{Input} The existing hyperedge list $H$, the newborn node $v$, and the selection $u$ of the RSS step \\
\noindent\hspace*{\algorithmicindent} \textbf{Output} The new state of the hyperedge list $H$
\begin{algorithmic}[1]
    \State initialize two new hyperedges $h_x$ and $h_y$
    \State divide the vertices of $u$ in $h_x$ and $h_y$
    \State add $v$ in $h_x$ and $v$ in $h_y$
    \State select $m-2$ hyperedges from $H$: $h_1, h_2, \dots, h_{m-2}$
    \ForAll{$i \in \langle 1,m-2 \rangle$}
        \State $h_c \gets $ a non-empty hyperedge in $\{h_x,h_y\}$
        \State find an element $w$ of $h_i$ not present in $h_c$
        \State add $w$ in $h_c$ and replace it with $v$ in $h_i$
    \EndFor
    \State add $h_x$ and $h_y$ on $H$
    \State \textbf{return} $H$
\end{algorithmic}
\label{alg:algorithm-2}
\end{algorithm}

First, the two new hyperedges $h_x$ and $h_y$ are initialized with the vertices of the randomly selected set $u$ from the RSS step. Although this choice does not impact the correctness of the algorithm and can be executed arbitrarily, a sensible option is a half split, or a near-half split if $|u|$ is odd. One copy of the newborn node $v$ is then added to each of the $h_x$ and $h_y$ hyperedges as it cannot have been previously contained in either. The algorithm then selects $m-2$ existing hyperedges to perform the swap of the $m-2$ remaining copies of $v$. The selection of these hyperedges is also irrelevant to the correctness of the algorithm and can be performed arbitrarily but it is generally desirable or useful to randomize the selection. One possible algorithm is shown in~\cite[Section II.B, Alg. 3]{batagelj2005efficient} that operates using a virtual shuffle and performs $m-2$ selections without the possibility of collisions. This algorithm has the property that all higher-order inclusion probabilities are equal, i.e. all possible $m-2$ groups of hyperedges are equiprobable to be selected.
The algorithm selects one node from each of these $m-2$ hyperedges that is not already present in either $h_x$ or $h_y$ and swaps its value with $v$. This selection can also be done in different ways but one option consistent with the choices made previously is to traverse an existing hyperedge in random order until a node is found not to be contained in the new hyperedge. Finally, the new hyperedges $h_x$ and $h_y$ are appended into $H$. Note that Algorithm SE-B reduces to Algorithm SE-A for $m=2$, as no exchanges are taking place.

Algorithm SE-B satisfies the correctness invariants, as during the updating of the $H$ hypergraph, the newborn node that has degree $m$ gains exactly $m$ hyperedges and each node in $u$, whose degree is increased by 1, gains one more hyperedge. No additional insertions or removals are performed, except swaps, leading to the following theorem:

\begin{theorem}
    Algorithm SE-B satisfies the \textit{str$\pi$ps} probability scheme and generates a simple graph according to the definition of the BA model.
\end{theorem}

The complexity of algorithm SE-B can be derived from the complexity of the \textit{update} function in Algorithm~\ref{alg:algorithm-2}. The initial steps are operations that can be performed in time proportional to $m$, including the selection of the $m-2$ existing hyperedges. In order to fill $h_x$ and $h_y$, there are
\[
    2 \cdot (m/2 + (m/2+1) + (m/2+2) + ... + (m-1)) = \Theta(m^2)
\]
operations required in the worst case when every element checked in the existing hyperedges except for the last exists in either $h_x$ or $h_y$. The complexity of the RSS step in Algorithm~\ref{alg:algorithm-abstract} is $\Theta(\rss m)$ in the worst case. This leads to the following theorem:

\begin{theorem}
    Algorithm SE-B runs in time $\altmathcal{O}(n m \rss + n m^2)$ to create a graph of $n$ vertices.
\end{theorem}

Despite algorithm SE-B running in linear time with respect to $n$, the complexity of the worst case is unlikely to occur in a typical instance of the problem. Given that as $n$ increases, the probability of collisions when finding elements in the existing hyperedges not already present in $h_x$ and $h_y$ is being reduced, we conjecture that the average case complexity of algorithm SE-B is $\altmathcal{O}(n m \rss)$. This hypothesis is supported by experimental observations but should be pursued in future work.

As a closing remark, similar to SE-A, an inherent property of Algorithm SE-B is the high-order sampling bias towards the structure of the graph. Specifically, groups of vertices that have been selected together in the past are more likely to also be selected together in the future, since the nodes inside $u$ are used to populate the new hyperedges. This behavior is sometimes desired, for example social networks are not typically uncorrelated and have some degree of underlying structure. In other cases, the parameter $\rss$ provides a small amount of control over these properties; the larger it is, the smaller the impact of the graph structure in the high-order sampling tends to be. Similar to Algorithm SE-A, as previously proven, Algorithm SE-B --and any implementation of Algorithm SE-- is guaranteed to have some amount of bias present. In contrast to SE-A, however, SE-B does not have the same guarantees regarding the strict positivity of the inclusion probability of all $m$-combinations.

\subsection{Algorithm SE-B*}

In this section, we present Algorithm SE-B*, a variation of Algorithm SE-B, that guarantees that every $m$-combination of nodes has a strictly positive joint inclusion probability for all steps of the generation process. This is achieved by formulating an additional invariant for the auxiliary data structure $H$, namely Invariant~3:

\begin{description}
    \item[Invariant~3]
        For every ordered pair of vertices $v,u$ with $v \neq u$, there is at least one hyperedge $h \in H$, such that $v \in h$ and $u \notin h$.
\end{description}

The invariant guarantees that every $m$-combination has non-zero inclusion probability via the following theorem:

\begin{theorem}
For any SE algorithm that satisfies Invariant~3 with $z \ge m$, the $m$-th order inclusion probabilities of any set of $m$ distinct nodes is strictly positive.
\end{theorem}
\begin{proof}
Consider an instance of the SE algorithm with $z=m$. Assume an arbitrary set of $m$ distinct nodes $v_1, v_2, \ldots, v_m$. Each node $v_i$ has at least $m$ hyperedges in the auxiliary hypergraph. Assume the hyperedges $h_1, h_2, \dots, h_m$ that are selected for the RSS step, such that $v_i \in h_i$. If there is another node $u_1$ that belongs to all the hyperedges $h_1, h_2, \dots, h_m$, then this node $u_1$ will have an inclusion probability of value $1$ in any sample of size $m$ and consequently belong to any RSS sample generated from these hyperedges. This would make the probability of the set $v_1, v_2, \ldots, v_m$ zero. However, because of Invariant~3, there is at least one hyperedge $h_1'$, such that $v_1 \in h_1'$ and $u_1 \notin h_1'$. In such case, the hyperedge $h_1'$ could have been selected in the RSS selection instead of $h_1$ in order for $u_1$ to appear at most $m-1$ times. The argument can be applied to any other node that appears $m$ times. In the end, no vertex $\notin v_1, v_2, \dots, v_m$ will belong to all hyperedges $h_1, h_2, \dots, h_m$. Applying $RSS$ to the final set of $m$ hyperedges gives with strictly positive probability the sample $v_1, v_2, \dots, v_m$. It is easy to see that the theorem holds for $z>m$ as well.
\end{proof}

Invariant~3 is indirectly satisfied in SE-A because there cannot be any parallel edges in the graph and, hence, even if nodes $u,v$ are adjacent, there must be an edge incident to $u$ but not $v$ since all nodes have degree greater or equal than 2. This situation, however, needs to be enforced for $m>2$ by adjusting the \textit{update} step. Constructing Algorithm SE-B* that satisfies this invariant must start from an initial graph that corresponds to an auxiliary hypergraph that already satisfies Invariant~3. For example, one such option would be the complete graph of $m+1$ vertices, which implies that the complete graph of $m$ vertices is also a valid initial graph, despite not satisfying Invariant~3. Algorithm~\ref{alg:algorithm-2-star} shows a high level sketch of the proposed SE implementation that satisfies Invariant~3.

\begin{algorithm} Algorithm SE-B* -- update function (high level sketch) \\
\noindent\hspace*{\algorithmicindent} \textbf{Input} The existing hyperedge list $H$, the newborn node $v$, and the selection $u$ of the RSS step \\
\noindent\hspace*{\algorithmicindent} \textbf{Output} The new state of the hyperedge list $H$
\begin{algorithmic}[1]
    \State initialize two new hyperedges $h_a$ and $h_b$
    \State $h_a \gets \{ u_1, u_2, \dots, u_{m-1}, v\}$
    \State $h_b \gets \{ u_m, v \}$
    \State select $m-2$ hyperedges from $H$: $h_1, h_2, \dots, h_{m-2}$
    \ForAll{$i \in \langle 1,m-2 \rangle$}
        \State find an element $w_i$ of $h_i$ not present in $h_b$
        \State add $w_i$ to $h_b$
    \EndFor
    \ForAll{$i \in \langle 1,m-2 \rangle$}
        \State \textit{ensure Invariant~3} for $w_i$ and get $h_r$ as the returned hyperedge
        \State replace $w_i$ with $v$ in $h_r$
    \EndFor
    \State add $h_a$ and $h_b$ to $H$
    \State \textbf{return} $H$
\end{algorithmic}
\label{alg:algorithm-2-star}
\end{algorithm}

\begin{algorithm} \textit{ensure Invariant~3} (high level sketch) \\
\noindent\hspace*{\algorithmicindent} \textbf{Input} The new hyperedge $h_b$, the hyperedge $h_i$ that was initially selected, and the node $w_i$ that was selected from $h_i$ \\
\noindent\hspace*{\algorithmicindent} \textbf{Output} The hyperedge that must take the place of $h_i$, or $h_i$ if no change is needed
\begin{algorithmic}[1]
    \State Let $h_i^{(j)}, j=1,\ldots,m$ be $m$ hyperedges in $H$ containing node $w_i$
    \State Assume $h_i^{(1)} = h_i$ (the selected hyperedge $h_i$ appears as $h_i^{(1)}$)
    \State Compute
        \State ~~~~ $s_2 = (h_b \cap h_i^{(2)}) - \{ w_m, v\}$
        \State ~~~~ $s_3 = (s_2 \cap h_i^{(3)})$
        \State ~~~~ $s_4 = (s_3 \cap h_i^{(4)})$
        \State ~~~~ \dots
        \State ~~~~ $s_m = (s_{m-1} \cap h_i^{(m)})$
        \State ~~~~ $s_{m+1} = s_m - h_i^{(1)}$ 
    \If{$s_{m+1} = \emptyset$}
        \State \textbf{return} $h_i$ (no change needed as Invariant~3 is satisfied)
    \Else
        \State There exists $k: 2 < k \leq m$, such that $s_k = s_{k-1}$.
        \State \textbf{return} $h_i^{(k)}$
    \EndIf
\end{algorithmic}
\label{alg:algorithm-Invariant-3}
\end{algorithm}

We make the following claim regarding the algorithm described in Algorithm~\ref{alg:algorithm-2-star}:

\begin{claim}
Algorithm SE-B* satisfies Invariant~3 and, hence, guarantees strictly positive inclusion probabilities of any $m$-combination of nodes on any step of the generation process for $z \ge m$.
\end{claim}

\begin{proof}
For each $i \in 1,2,\dots,(m-2)$, let
\begin{align*}
    {\theta}_1 &= (h_b - h_i^{(k)}) - \{ w_m, v\} \\
    {\theta}_2 &= (z_1 \cap h_i^{(2)}) \\
    {\theta}_3 &= (z_2 \cap h_i^{(3)}) \\
    & \ldots \\
    {\theta}_{k-1} &= (z_{k-2} \cap h_i^{(k-1)})
\end{align*}
The intuition of the approach is that the sets ${\theta}_i$ contain all nodes $w_x \in h_b$ for which there is a possibility that the pair $w_i, w_x$ violates Invariant~3. Note that ${\theta}_{k-1}$ is equal to $s_{k-1} - h_m^{(k)} = \emptyset$. Therefore Invariant 3 is satisfied. This holds for each hyperedge $h_i$, for $i \in 1,2,\dots,(m-2)$.
\end{proof}

Regarding the complexity of Algorithm SE-B*, for each new node inside the \textit{update} step, the algorithm randomly selects $m-2 = \altmathcal{O}(m)$ hyperedges from the auxiliary hypergraph. For each hyperedge, at most $m-2 = \altmathcal{O}(m)$ set intersections are performed. Finally, each set intersection is performed between two sets of at most $m$ nodes each. Overall, the complexity is $\altmathcal{O}(m^3)$. Consequently, considering the extra computational step of RSS:

\begin{theorem}
    Algorithm SE-B* runs in time $\altmathcal{O}(nmz + n m^3)$ to create a graph of $n$ vertices.
\end{theorem}

Although Algorithm SE-B* guarantees non-zero inclusion probabilities for all $m$-combinations of vertices given that $z \ge m$, in practice this property might not result in a significantly different algorithm than SE-B. Preliminary experimental evidence suggests that Invariant~3 is only rarely violated in SE-B and, even then, it might only affect a single pair of vertices. A violation of Invariant~3 is even more rare as $n$ increases. The practical differences between SE-B and SE-B* should be investigated in future work.

\section{Algorithm SE-C}
\label{sec:algorithm-3}

Algorithm SE-C is another implementation of the abstract SE algorithm and demonstrates an alternative way to distribute the copies of $v$ and $u$ inside the $H$ auxiliary data structure more randomly than Algorithm SE-B.
The core idea of algorithm SE-C is, instead of swapping only a single element of each of the $m-2$ existing hyperedges to replace $v$, to shuffle all copies inside the $m-2$ existing hyperedges as well as $h_x$ and $h_y$ ($m^2$ elements in total). Random systematic partitioning (Algorithm~\ref{alg:random-systematic-partitioning}) fits this concept perfectly, as it is guaranteed that no vertex will have more copies than the number of hyperedges ($v$ being the maximum with $m$ copies).

A high level sketch of the \textit{update} function of algorithm SE-C is shown in Algorithm~\ref{alg:algorithm-3}. Following the scheme of random systematic partitioning, the contents of the $m-2$ existing hyperedges and $u$ as well as $m$ copies of $v$ are inserted into the partitioning algorithm and the resulting groups are replacing their old records, comprising the new value of the $H$ hypergraph. It is worth noting that the shuffling performed by algorithm SE-C can be tuned by increasing the number of hyperedges inserted into random systematic partitioning, for example $2m$ ($2$ new and $2m-2$ existing) instead of $m$.

\begin{algorithm} Algorithm SE-C -- update function (high level sketch) \\
\noindent\hspace*{\algorithmicindent} \textbf{Input} The existing hyperedge list $H$, the newborn node $v$ and the selection $u$ of the RSS step \\
\noindent\hspace*{\algorithmicindent} \textbf{Output} The new state of the hyperedge list $H$
\begin{algorithmic}[1]
    \State initialize empty array $A$
    \State select and remove $m-2$ hyperedges from $H$: $h_1, h_2, \dots, h_{m-2}$
    \State add $m \cdot (m-2)$ elements from $h_1, h_2, \dots, h_{m-2}$ into $A$
    \State add $m$ copies of $v$ into $A$
    \State add all elements of $u$ into $A$
    \State perform random systematic partitioning on $A$ with $s=m$
    \State add the $m$ sets of $A$ to $H$
    \State \textbf{return} $H$
\end{algorithmic}
\label{alg:algorithm-3}
\end{algorithm}

Algorithm SE-C performs the same amount of additions as Algorithm SE-B and SE-B*. Therefore, there is no change in the correctness in relation to algorithm SE-B:

\begin{theorem}
    Algorithm SE-C satisfies the \textit{str$\pi$ps} probability scheme and generates a simple graph according to the definition of the BA model.
\end{theorem}

Regarding its complexity, the updating step of Algorithm SE-C is bounded by the random systematic partitioning that is applied on $m^2$ elements and, hence, considering Theorem~\ref{theorem:random-systematic-partitioning-complexity}, contributes to the running time with a proportionality of $m^2$. Unlike Algorithm SE-B, the running time of the updating step is not impacted by the random number generator and its asymptotic performance is always proportional to $m^2$. The RSS step is identical to the SE-B algorithm, resulting to the following statement regarding the overall running time of Algorithm SE-C:

\begin{theorem}
    Algorithm SE-C runs in time $\altmathcal{O}(n m \rss + n m^2)$ to create a graph of $n$ vertices.
\end{theorem}

Unlike SE-B*, Algorithm SE-C is not guaranteed to satisfy Invariant 3 and, as a result, does not necessarily guarantee strictly positive joint inclusion probabilities for all combinations of vertices and for all time steps. Experimental observations, however, suggest that this situation is very rare in practice when using $\rss \ge m$.

As a closing remark, it is interesting to note that algorithm SE-C demonstrates the close association between preferential attachment and the random sampling problem. In fact, five different random sampling methods are involved in the design of algorithm SE-C to solve the preferential attachment problem:
\begin{enumerate}
    \item A sampling algorithm to select $\rss$ hyperedges from the population of hyperedges in $H$ with replacement before the RSS step.
    \item Random systematic sampling for the random selection of the vertices that receive connections from the newborn node.
    \item One sampling algorithm to select $m-2$ existing hyperedges without replacement from $H$ before the RSP step.
    \item Random systematic partitioning, which is influenced by systematic random sampling and is used to both initialize the $H$ array from $\altmathcal{G}_0$ and to rearrange the node copies inside the $m$ hyperedges.
    \item Jessen's whole sampling method to update the $H$ array efficiently, enforcing inclusion probabilities strictly proportional to the degree of the vertices.
\end{enumerate}
In this paper, we exploit this inherent relation in order to develop an implementation of the growing preferential attachment mechanism that perfectly fits the requirements of the \textit{str$\pi$ps} probability scheme. Future integration between these two problems should also be pursued in the future.

\section{Discussion: The Initial Graph}
\label{sec:initial-graph}

In this section, we discuss the options for the initial graph $\altmathcal{G}_0(V_0,E_0)$ that can be used in our algorithms, explain its requirements that were previously omitted for brevity, and propose methods to satisfy them. 

The global features of the scale free graph are expected to be independent of the initial graph, as the BA model is typically regarded a stationary distribution model. However, the initial graph state influences features of the first nodes, for example the probability that a specific node will become the heaviest node in the social network.
For consistency and completeness, we note that the initial graph constitutes a state of our methods (the first state) and, as such, needs to satisfy the invariants of Section~\ref{sec:algorithm-abstract} in order for the transformation into the hypergraph $H$ to be feasible. In particular:
\begin{enumerate}
    \item The number of edges $|E_0|$ in the initial graph needs to satisfy the divisibility $2|E_0|/m$.
    \item No vertex can have degree higher than $2|E_0|/m$.
\end{enumerate}
It is worth noting that the complete graph of $m$ vertices, which is a typical initial graph used in the BA model, satisfies both of these conditions without any processing required.

Regarding requirement (2), and assuming that requirement (1) is satisfied, a vertex may not have degree that is bigger than $2|E_0|/m$, otherwise the number of hyperedges in the $H$ data structure will not be enough for the copies of this vertex; at least one hyperedge would have to contain multiple copies, which is not allowed. For example, in the star graph of 5 vertices and $m=4$, the center node has degree 4 while the sum of the degrees is 8. Thus, there are 2 hyperedges in $H$ but the center node needs to have 4 copies in $H$, which is impossible. This situation highlights the inherent issue of infeasibility in random sampling when the \textit{str$\pi$ps} model is used~\cite{efraimidis2015weighted}. In the previous example, the inclusion probability of the center node is $(4/8) \cdot 4 = 2$ (200\%). A straightforward approach is to accept the fact that the probabilities are infeasible and to bound all infeasible probabilities to be at most $1$, until the probabilities gradually become feasible as the number of nodes $n$ increases. We do not further discuss this issue.

Regarding requirement (1), considering that $H$ needs to contain an integer amount of hyperedges, it follows that $2|E_0|$ needs to be divisible by $m$. For example, a complete graph of $6$ nodes for $m=4$ does not satisfy this property (30 node copies are not divisible by 4). In the rest of this section, we discuss two methods to address the limitations imposed by requirement (1), namely forcing the number of edges to a specific value that does not oppose the requirement and introducing a multiplication factor that enlarges the count of all entries of the problem.

\paragraph{Forcing the number of edges}

Initially, $\altmathcal{G}_0$ can be transformed into a graph that satisfies requirement (1) by selecting an appropriate number of edges and using the $\altmathcal{G}(n,M)$ generator to produce the initial graph. The minimum number of edges in the initial graph such that it satisfies the requirement is
\[
    |E_0|_{min} = \frac{\text{lcm}(m,2)}{2},
\]
while the largest number of edges depends on $|V_0|$ and is
\[
    |E_0|_{max} = \left\lfloor \frac{|V_0| \cdot (|V_0| - 1)}{\text{lcm}(m,2)} \right\rfloor \cdot \text{lcm}(m,2).
\]
Thus, it follows that
\[
    |V_0| \cdot (|V_0| - 1) - \text{lcm}(m,2) \le |E_0|_{max} \le |V_0| \cdot (|V_0| - 1),
\]
which implies that $|E_0|_{max}$ is within a margin of $m$ or $2m$ of the edges of the complete graph with the same number of vertices that is often used as input. Naturally, the $\altmathcal{G}(n,M)$ generator is still subject to requirement (2) and rejections should be used to ensure that.

\paragraph{Introducing a multiplication factor}

Another way to address the limitation caused by requirement (1) is to setup a factor of multiplication for the entire process. The multiplication factor is
\[
    \lambda = \frac{\text{lcm}(2|E_0|,m)}{2|E_0|}
\]
and denotes the factor with which all node copies are multiplied with. Hence, for the initial graph, instead of inserting $2|E_0|$ entries in $H$, which might not be divisible by $m$, we are inserting $2 \lambda |E_0| = \text{lcm}(2|E_0|,m)$ entries, which is divisible by $m$. Similarly, for the duration of the process, instead of inserting 2 hyperedges, we insert $2 \lambda$ hyperedges, where the copies of the vertices are multiplied also by $\lambda$. While our algorithms are easy to be generalized to support this process and completely solve the limitation if such behavior is desired, this method costs in memory and design complexity and for most cases the solution of generating an initial $\altmathcal{G}(n,M)$ graph with the closest number of acceptable edges should be preferred.

\section{Degree Distribution on Computer Experiments}

\begin{figure}
    \centering
    \includegraphics[width=\textwidth]{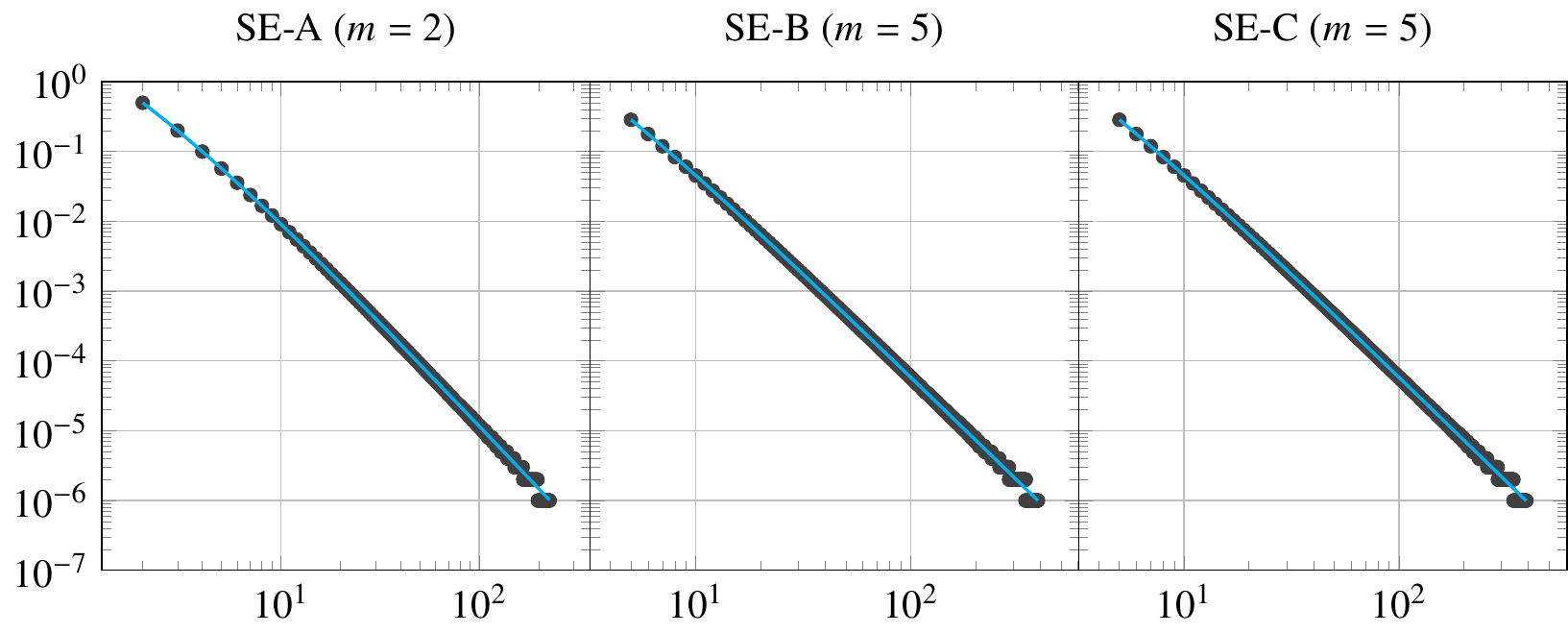}
    \caption{Degree distribution for the SE-A, SE-B and SE-C algorithms in a log-log plot for $n=300\,000$. For the SE-A algorithm it is $m=2$ while for SE-B and SE-C it is $m=5$. The cyan lines that are rendered on top of the marks show the theoretical degree distribution. The close association between the theoretical expectation and the experimental models can be observed.}
    \label{fig:degree-distribution}
\end{figure}

Figure~\ref{fig:degree-distribution} shows the experimental degree distributions of the three algorithms; from left to right SE-A, SE-B and SE-C. For the experimental approach of the degree distribution of our algorithms we use $n=300\,000$ in order to capture an approximation of the asymptotic state of the limiting distribution. By definition, the SE-A algorithm is only compatible with $m=2$. For the SE-B and SE-C algorithms we use $m=5$. In terms of the initial graph, the complete graph of $|V_0|=m$ is used. Regarding the options of the methods, the version of the SE-C algorithm refers to the algorithm that shuffles $m$ hyperedges (instead of more), while the split of $u$ in the two new hyperedges is performed using the random method with equal split.

The plots were generated using 10\,000 iterations of the generation process to achieve statistical stability. The plots also contain the theoretical degree distribution (Equation~\ref{eq:degree-distribution}) as a cyan line rendered on top of the data points. We note that, despite the theoretical distribution being a discrete probability distribution, it is rendered here as continuous in order to be visually distinguishable from the data points. The simulation shows that the resulting graphs are scale-free and an almost exact fit with the theoretical distribution. Although algorithms SE-A, SE-B and SE-C are different in their operation and their internal preferential attachment mechanism, they all result in scale-free distribution because they all satisfy the \textit{str$\pi$ps} property, as otherwise proven in~\cite{albert2002statistical}.

The reference implementation of our algorithms that was used for the plots is available online\footnote{https://github.com/gstamatelat/se}.

\section{Conclusions}

In this paper, we have utilized a series of random sampling schemes and methods to design a class of scale-free graph generator algorithms. Our models obey the dynamics of the preferential attachment scheme and the definition of the BA model. In particular, the algorithms are designed such that the inclusion probability of any vertex and at any time of the process is exactly proportional to its degree. This behavior is in contrast to existing methods where the inclusion probabilities are only approximately proportional to their degree. Our algorithms make use of precise and diverse random sampling steps and run in linear time with respect to the desired graph size $n$ for any fixed $m$. This is, to our knowledge, the first time that both strict probability interpretation and linear complexity are achieved for the generation of scale-free graphs via the preferential attachment mechanism.

Our analysis started with Algorithm SE-A, the simple case for $m=2$ that enriches existing models in the literature with clustering analysis and joint inclusion options. The correctness of this algorithm was shown by the fact that a vertex exists in the edge set as many times as its degree. The generalization of SE-A, algorithm SE, is proposed, that uses an auxiliary data structure $H$ that resembles a $m$-uniform hypergraph and works for $m \ge 2$. While the operation of this algorithm is abstract, the necessary invariants are defined that guarantee its correctness. Finally, algorithms SE-B and SE-C implement algorithm SE more concretely by either exchanging the necessary values in $H$ to ensure the invariants are satisfied or completely shuffling parts of $H$ respectively. Algorithm SE-B* is a variation of SE-B that additionally guarantees that every $m$-combination of nodes has a non-zero probability of being jointly selected at any step of the generation process. Overall, our analysis raises interesting future work directions about the higher-order properties surrounding the clustering features of the resulting scale-free graphs.

\bibliographystyle{unsrt}
\bibliography{main}

\end{document}